\documentclass[prl,superscriptaddress,twocolumn]{revtex4}
\usepackage{color}
\usepackage{amsfonts,amssymb,amsmath}
\usepackage[]{graphics,graphicx,epsfig}
\usepackage{amsthm}
\usepackage{enumitem}
\usepackage{bbm}
\usepackage{dsfont}

\usepackage{braket}
\usepackage{amsmath,amsthm,amsfonts,amssymb,amscd}
\usepackage[utf8]{inputenc}
\usepackage{indentfirst}
\usepackage{graphicx}
\usepackage[T1]{fontenc}
\usepackage{mathpazo}
\usepackage[USenglish]{babel}
\usepackage[colorlinks=true]{hyperref}
\usepackage{hyperref}
\usepackage{bm}
\usepackage{lscape}
\def\id{\leavevmode\hbox{\small1\kern-3.8pt\normalsize1}}

\def\identity{\leavevmode\hbox{\small1\kern-3.8pt\normalsize1}}

\renewcommand{\epsilon}{\varepsilon}

\newtheorem{definition}{Definition} 

\newtheorem{thm}[definition]{Theorem}
\newtheorem{corollary}[definition]{Corollary}

\newtheorem*{rep@theorem}{\rep@title}
\newcommand{\newreptheorem}[2]{%
\newenvironment{rep#1}[1]{%
 \def\rep@title{#2 \ref{##1} (restatement)}%
 \begin{rep@theorem}}%
 {\end{rep@theorem}}}
\makeatother

\newreptheorem{thm}{Theorem}
\newreptheorem{lem}{Lemma}

\def\ba#1\ea{\begin{align}#1\end{align}}
\def\ban#1\ean{\begin{align*}#1\end{align*}}

\newcommand{\be}{\begin{equation}}
\newcommand{\ee}{\end{equation}}

\def\benum{\begin{enumerate}}
\def\eenum{\end{enumerate}}

\def\squareforqed{\hbox{\rlap{$\sqcap$}$\sqcup$}}
\def\qed{\ifmmode\squareforqed\else{\unskip\nobreak\hfil
\penalty50\hskip1em\null\nobreak\hfil\squareforqed
\parfillskip=0pt\finalhyphendemerits=0\endgraf}\fi}
\def\endenv{\ifmmode\;\else{\unskip\nobreak\hfil
\penalty50\hskip1em\null\nobreak\hfil\;
\parfillskip=0pt\finalhyphendemerits=0\endgraf}\fi}

\newcommand{\<}{\langle}
\renewcommand{\>}{\rangle}

\def\id{{\operatorname{id}}}

\def\be{\begin{equation}}
\def\ee{\end{equation}}
\def\ben{\begin{eqnarray}}
\def\een{\end{eqnarray}}

\def\bei{\begin{itemize}}
\def\eei{\end{itemize}}

\mathchardef\ordinarycolon\mathcode`\:
\mathcode`\:=\string"8000
\def\vcentcolon{\mathrel{\mathop\ordinarycolon}}
\begingroup \catcode`\:=\active
  \lowercase{\endgroup
  \let :\vcentcolon
  }

\newcommand{\nc}{\newcommand}
 \nc{\proj}[1]{|#1\rangle\!\langle #1 |} 
\nc{\avg}[1]{\langle#1\rangle}

\nc{\conv}{\operatorname{conv}}
\nc{\smfrac}[2]{\mbox{$\frac{#1}{#2}$}} \nc{\Tr}{\operatorname{Tr}}
\nc{\ox}{\otimes} \nc{\dg}{\dagger} \nc{\dn}{\downarrow}
\nc{\lmax}{\lambda_{\text{max}}}
\nc{\lmin}{\lambda_{\text{min}}}

\nc{\csupp}{{\operatorname{csupp}}}
\nc{\qsupp}{{\operatorname{qsupp}}} \nc{\var}{\operatorname{var}}
\nc{\rar}{\rightarrow} \nc{\lrar}{\longrightarrow}
\nc{\poly}{\operatorname{poly}}
\nc{\polylog}{\operatorname{polylog}} \nc{\Lip}{\operatorname{Lip}}
\nc{\Om}{\Omega}
\nc{\wt}[1]{\widetilde{#1}}

\def\>{\rangle}
\def\<{\langle}

\def\d{\delta}

\nc{\glneq}{{\raisebox{0.6ex}{$>$}  \hspace*{-1.8ex} \raisebox{-0.6ex}{$<$}}}
\nc{\gleq}{{\raisebox{0.6ex}{$\geq$}\hspace*{-1.8ex} \raisebox{-0.6ex}{$\leq$}}}

\nc{\vholder}[1]{\rule{0pt}{#1}}
\nc{\wh}[1]{\widehat{#1}}
\nc{\h}[1]{\widehat{#1}}

\nc{\ob}[1]{#1}

\def\beq{\begin {equation}}
\def\eeq{\end {equation}}

\def\be{\begin{equation}}
\def\ee{\end{equation}}

\nc{\eq}[1]{(\ref{eq:#1})} 
\nc{\eqs}[2]{\eq{#1} and \eq{#2}}

\nc{\eqn}[1]{Eq.~(\ref{eqn:#1})}
\nc{\eqns}[2]{Eqs.~(\ref{eqn:#1}) and (\ref{eqn:#2})}

\nc{\region}{\cS\cW}

\newenvironment{protocol*}[1]
  {
    \begin{center}
      \hrulefill\\
      \textbf{#1}
  }
  {
    \vspace{-1\baselineskip}
    \hrulefill
    \end{center}
  }

\begin{document}
\title{All two-party facet Bell inequalities are violated by Almost Quantum correlations}
\author{Ravishankar \surname{Ramanathan}}
\email{ravishankar.r.10@gmail.com}
\affiliation{Department of Computer Science, The University of Hong Kong, Pokfulam Road, Hong Kong}
\begin{abstract}
The characterization of the set of quantum correlations is a problem of fundamental importance in quantum information. The question whether every proper (tight) Bell inequality is violated in Quantum theory is an intriguing one in this regard. Here, we make significant progress in answering this question, by showing that every tight Bell inequality is violated by 'Almost Quantum' correlations, a semi-definite programming relaxation of the set of quantum correlations. As a consequence, we show that many (classes of) Bell inequalities including two-party correlation Bell inequalities and multi-outcome non-local computation games, that do not admit quantum violations, are not facets of the classical Bell polytope. To do this, we make use of the intriguing connections between Bell correlations and the graph-theoretic Lov\'{a}sz-theta set, discovered by Cabello-Severini-Winter (CSW). We also exploit connections between the cut polytope of graph theory and the classical correlation Bell polytope, to show that correlation Bell inequalities that define facets of the lower dimensional correlation polytope are violated in quantum theory. The methods also enable us to derive novel (almost) quantum Bell inequalities, which may be of independent interest for self-testing applications.
\end{abstract}

\maketitle


Quantum correlations, i.e., the correlations between quantum systems in a Bell-type experiment, are of central interest in Quantum Information Theory. Their violation of Bell inequalities shows, in a device-independent manner, that quantum theory fundamentally differs from all classical theories. These quantum 'non-local' correlations also allow to perform tasks that are impossible in classical theory, such as generation of cryptographic key secure against post-quantum eavesdroppers \cite{Ekert}, intrinsic randomness certification and amplification \cite{Buhrman10}, and reduction of communication complexity \cite{Pironio10, Brandao16}. For fundamental reasons as well as to develop these applications, it is of utmost importance to characterize the set of quantum correlations, and understand how it fits in between the polytopes of classical and general non-signalling correlations. 

The proper (tight) Bell inequalities are facets of the classical polytope, that are not also facets of the no-signalling polytope. A problem of fundamental importance in the characterization of quantum correlations was raised by Gill in \cite{Gill05}, namely whether every tight Bell inequality is violated in quantum theory. The analogous question pertaining to facets of the binary-outcome correlation polytope (the classical polytope of two-party binary-outcome correlations, excluding the local marginal terms) was raised by Avis et al. in \cite{AII06}. Escol\'{a}, Calsamiglia and Winter in a recent breakthrough result \cite{ECW20} answered the latter question, showing that the binary-outcome correlation polytope does not share any non-trivial facets in common with the set of quantum correlations. The corresponding question for multi-party tight Bell inqualities had been previously answered in a fundamental breakthrough result by Fritz et al. who identified a class of non-trivial tight Bell inequalities (called Local Orthogonality inequalities) that are not violated in quantum theory, when three or more parties are involved in the Bell experiment. The corresponding Local Orthogonality principle is a fundamental information-theoretic principle that serves to delineate the set of correlations realizable in a physical theory. In the multipartite Bell scenario, it only remains an open question whether all non-trivial facet Bell inequalities without quantum violation are of the local orthogonality form. In the bipartite Bell experiment though, the local orthogonality constraints reduce to the no-signalling conditions, and do not provide any non-trivial facet constraints to the quantum correlation set. 

In this paper, we study the question of whether there are tight two-party Bell inequalities with no quantum violation. We first describe novel classes of two-party Bell inequalities that do not admit quantum violation, including certain two-party correlation Bell inequalities and their multi-outcome generalization. We then prove our central result that all two-party Bell inequalities that define facets of the classical Bell polytope, are violated by a natural semi-definite programming relaxation to the set of quantum correlations, that has been dubbed 'Almost Quantum' theory. We show how the novel classes introduced earlier are proven to not describe facets as a consequence of the result, as well as how it subsumes a recent breakthrough  on two-party XOR games with no quantum advantage. We also show how connections discovered by Avis et al. in \cite{AII06} can be used to show that all correlation Bell inequalities that define facets of the lower dimensional correlation Bell polytope, are violated in quantum theory. Finally, we show how the methods used in proving these results also enable us to derive novel quantum Bell inequalities, a fact which has potential for applications to self-test quantum correlations \cite{Scarani}. We end with a brief discussion and open questions.

\subsection{Preliminaries}

Consider a two-party Bell experiment. Suppose one party Alice chooses to measure one of $m_A$ inputs $i_A = 1, \dots, m_A$, and obtains one of $d_{A, i_A}$ outputs $o_A \in \{1, \dots, d_{A, i_A}\}$. Similarly, the other party Bob chooses to measure one of $m_B$ inputs $i_B = 1, \dots, m_B $, and obtains one of $d_{B, i_B}$ outputs $o_B \in \{1, \dots, d_{B, i_B}\}$ outputs. Such a Bell scenario is denoted by the notation $\textbf{B}(2; m_A, \vec{d}_A; m_B, \vec{d}_B)$, where $\vec{d}_A = (d_{A,1}, \dots, d_{A, m_A})$ and $\vec{d}_B = (d_{B,1}, \dots, d_{B, m_B})$. In some instances, this notation can also be shortened to $\textbf{B}(\vec{d}_A, \vec{d}_B)$ for simplicity. The joint probability of obtaining the outcomes $(o_A, o_B)$ given the measurement settings $(i_A, i_B)$ is denoted as $P_{O_A, O_B | I_A, I_B}(o_A, o_B | i_A, i_B)$. We will view these $n = \left(\sum_{i_A = 1}^{m_A} d_{A, i_A} \right) \left(\sum_{i_B = 1}^{m_B} d_{B, i_B} \right)$ probabilities as forming the components of a vector $P_{O_A, O_B | I_A, I_B} = | P \rangle$ in $\mathbb{R}^n$, where the inputs and outputs are implicit, and the probabilities are also described as forming a box $P$.  

In the Bell scenario $\textbf{B}\left((2,2), (2,2)\right)$ where each party chooses $2$ dichotomic observables, all the facet inequalities are known: up to permutation of the outcomes they correspond to the well-known Clauser-Horne-Shimony-Holt inequality \cite{CHSH}. While it is in principle possible using specific facet-enumeration algorithms to obtain all the facet inequalities of the classical polytope in any given Bell scenario, in practice the complexity of the problem grows rapidly and facet inequalities have been found in cases with a few more observables and outcomes \cite{Pironio2}. In fact, the problem of listing all facet Bell inequalities has been demonstrated to be NP-complete \cite{Avis} making this an important but hard-to-solve problem in the theory of quantum non-locality. 

The box $P$ is a valid normalized no-signalling box, satisfying the no-signalling constraints of relativity and the normalization of probabilities, if it obeys the constraints of
\begin{enumerate}
\item Non-negativity: $P_{O_A, O_B | I_A, I_B}(o_A, o_B | i_A, i_B) \geq 0$ for all $o_A, o_B, i_A, i_B$, 

\item Normalization: $\sum_{o_A = 1}^{d_{A, i_A}} \sum_{o_B = 1}^{d_{B, i_B}} P_{O_A, O_B | I_A, I_B}(o_A, o_B | i_A, i_B)  = 1$ for all $i_A, i_B$,

\item No-Signalling: 
\begin{widetext}
\begin{eqnarray}
\sum_{o_A = 1}^{d_{A, i_A}} P_{O_A, O_B | I_A, I_B}(o_A, o_B | i_A, i_B) &=& \sum_{o'_A = 1}^{d_{A, i'_A}} P_{O_A, O_B | I_A, I_B}(o'_A, o_B | i'_A, i_B), \quad \text{for all} \; \; i_A, i'_A, o_B, i_B, \nonumber \\
\sum_{o_B = 1}^{d_{B, i_B}} P_{O_A, O_B | I_A, I_B}(o_A, o_B | i_A, i_B) &=& \sum_{o'_B = 1}^{d_{B, i'_B}} P_{O_A, O_B | I_A, I_B}(o_A, o'_B | i_A, i'_B), \quad \text{for all} \; \; i_B, i'_B, o_A, i_A.
\end{eqnarray}
\end{widetext}
\end{enumerate}

The convex hull of all boxes $P$ that satisfy the above constraints forms the No-Signalling Polytope of the Bell scenario $\textbf{NS}\left[\textbf{B}(2; m_A, \vec{d}_A; m_B, \vec{d}_B)\right]$. A fundamental result in polyhedral theory, known as the Minkowski-Weyl theorem, states that a polytope represented as the convex hull of a finite number of points, such as $\textbf{NS}\left[\textbf{B}(2; m_A, \vec{d}_A; m_B, \vec{d}_B)\right]$ can also be equivalently represented as the intersection of finitely many half-spaces. One may write the above constraints in the form of inequalities, with the normalization and no-signalling equalities being written as two inequalities, and re-write the No-Signalling Polytope in the following canonical form:
\begin{equation}
\textbf{NS}\left[\textbf{B}(2; m_A, \vec{d}_A; m_B, \vec{d}_B)\right] = \Big\{ | P \rangle \in \mathbb{R}^n : A \cdot | P \rangle \leq | b \rangle \Big\}.
\end{equation}
Here, the matrix $A$ is an $m \times n$ matrix, with $m = n + 2 m_A m_B + 2 (m_A-1) \sum_{i_B = 1}^{m_B} \left(d_{B, i_B} - 1 \right) +2 (m_B - 1) \sum_{i_A = 1}^{m_A} \left(d_{A, i_A} - 1 \right)$. This value for $m$ comes from $n$ non-negativity constraints, $m_A m_B$ normalization equalities, and $(m_A-1) \sum_{i_B = 1}^{m_B} \left(d_{B, i_B} - 1 \right)$ no-signalling equalities defining Bob's marginal probabilities and similarly $(m_B - 1) \sum_{i_A = 1}^{m_A} \left(d_{A, i_A} - 1 \right)$ no-signaling equalities defining Alice's marginal probabilities. The vector $| b \rangle$ is an appropriate defined $m$-dimensional vector with entries in $\{0, 1, -1\}$. Crucially, this gives the dimensionality of the No-Signalling Polytope to be
\begin{widetext}
\begin{equation}
\text{dim} \left(\textbf{NS}\left[\textbf{B}(2; m_A, \vec{d}_A; m_B, \vec{d}_B)\right] \right) = \left( \sum_{i_A = 1}^{m_A} \left(d_{A, i_A} - 1 \right) + 1 \right) \left( \sum_{i_B = 1}^{m_B} \left(d_{B, i_B} - 1 \right) + 1 \right) - 1 =: D.
\end{equation}
\end{widetext}

The boxes within the No-Signalling polytope that additionally satisfy the integrality constraint given by
\begin{enumerate}
\setcounter{enumi}{3}
\item Integrality $P_{O_A, O_B | I_A, I_B}(o_A, o_B | i_A, i_B) \in  \{0,1\}$ for all $o_A, o_B, i_A, i_B$,
\end{enumerate}
are said to be Local Deterministic Boxes (LDBs). The convex hull of these LDBs forms the classical or Bell polytope denoted by $\textbf{C}\left[\textbf{B}(2; m_A, \vec{d}_A; m_B, \vec{d}_B)\right]$. This is the set of all correlations obtainable from local hidden variable theories. 

The set of Quantum Correlations denoted by $\textbf{Q}\left[\textbf{B}(2; m_A, \vec{d}_A; m_B, \vec{d}_B)\right]$ also lies within the No-Signalling polytope. This set consists of boxes $P$ where each component $P_{O_A, O_B | I_A, I_B}(o_A, o_B | i_A, i_B)$ is obtained as:
\begin{equation}
\label{eq:qcorr-def}
P_{O_A, O_B | I_A, I_B}(o_A, o_B | i_A, i_B) = \text{Tr}\left[ \rho \left(E^{A}_{i_A, o_A} \otimes E^{B}_{i_B, o_B} \right) \right]
\end{equation}
for some quantum state $\rho \in \mathcal{H}_d$ of some arbitrary dimension $d$, and sets of projection operators $\{E^{A}_{i_A, o_A}\}$ for Alice and $\{E^{B}_{i_B, o_B}\}$ for Bob. Notably, the measurement operators satisfy the requirements of (i) Hermiticity: $\left(E^{A}_{i_A, o_A}\right)^{\dagger} = E^{A}_{i_A, o_A}$ for all $i_A, o_A$, and $\left(E^{B}_{i_B, o_B}\right)^{\dagger} = E^{B}_{i_B, o_B}$, for all $i_B, o_B$, (ii) Orthogonality: $E^{A}_{i_A, o_A} E^{A}_{i_A, o'_A} = \delta_{o_A, o'_A} E^{A}_{i_A, o_A}$ for all $i_A$, and $E^{B}_{i_B, o_B} E^{B}_{i_B, o'_B} = \delta_{o_B, o'_B} E^{B}_{i_B, o_B}$ for all $i_B$, and (iii) Completeness: $\sum_{o_A} E^{A}_{i_A, o_A} = \mathds{1}$ for all $i_A$ and $\sum_{o_B} E^{B}_{i_B, o_B} = \mathds{1}$ for all $i_B$. The set $\textbf{Q}\left[\textbf{B}(2; m_A, \vec{d}_A; m_B, \vec{d}_B)\right]$ is convex but not a polytope. We have the inclusions $\textbf{C}\left[\textbf{B}(2; m_A, \vec{d}_A; m_B, \vec{d}_B)\right] \subseteq \textbf{Q}\left[\textbf{B}(2; m_A, \vec{d}_A; m_B, \vec{d}_B)\right] \subseteq \textbf{NS}\left[\textbf{B}(2; m_A, \vec{d}_A; m_B, \vec{d}_B)\right]$.

By the Minkowski-Weyl theorem, the set $\textbf{C}\left[\textbf{B}(2; m_A, \vec{d}_A; m_B, \vec{d}_B)\right]$ can also be equivalently represented as the intersection of finitely many half-spaces
\begin{widetext}
\begin{equation}
\textbf{C}\left[\textbf{B}(2; m_A, \vec{d}_A; m_B, \vec{d}_B)\right] = \Big\{ | P \rangle \in \mathbb{R}^n : B_{G_i} \cdot | P \rangle \leq \omega_c(G_i) \; \; \; \; \forall i \in I \Big\},
\end{equation}
\end{widetext}
where $\{B_{G_i} \cdot | P \rangle \leq \omega_c(G_i), \; \; i \in I \}$ is a finite set of inequalities. The inequalities supporting facets of $\textbf{C}\left[\textbf{B}(2; m_A, \vec{d}_A; m_B, \vec{d}_B)\right]$ provide a minimal set of such inequalities, and are usually referred to as facet Bell inequalities, or in some instances in the literature just as the Bell inequalities. In particular, any valid inequality for $\textbf{C}\left[\textbf{B}(2; m_A, \vec{d}_A; m_B, \vec{d}_B)\right]$ can be derived from the facet inequalities. 

The introduction of a few notions from polytope theory is in order here. Boxes $P_1, \dots, P_m$ in $\mathbb{R}^n$ are said to be affinely independent if the unique solution to $\sum_{i=1}^{m} \mu_i P_i = 0$, $\sum_{i=1}^{m} \mu_i = 0$ is that $\mu_i = 0$ for all $i =1, \dots, m$. Equivalently, the boxes are affinely independent if $P_2 - P_1, \dots, P_m - P_1$ are linearly independent. The affine hull of a set of boxes is the set of all their affine combinations. The affine set has dimension $K$, if the maximum number of affinely independent boxes it contains is $K+1$. An inequality $B_{G_i} \cdot | P \rangle \leq \omega_c(G_i)$ satisfied by all boxes in $\textbf{C}\left[\textbf{B}(2; m_A, \vec{d}_A; m_B, \vec{d}_B)\right]$ is called a valid Bell inequality. Given a valid inequality $B_{G_i} \cdot | P \rangle \leq \omega_c(G_i)$, the set 
\begin{equation}
F = \Big\{ | P \rangle \in \mathbb{R}^n :  B_{G_i} \cdot | P \rangle = \omega_c(G_i) \Big\}
\end{equation}
is called a face of the classical polytope and the inequality is said to support $F$. The dimension of $F$ is the dimension of its affine hull. If $F \neq \null$ and $F \neq \textbf{C}\left[\textbf{B}(2; m_A, \vec{d}_A; m_B, \vec{d}_B)\right]$, it is a proper face. Proper faces satisfy by definition $dim(F) \leq dim\left(\textbf{C}\left[\textbf{B}(2; m_A, \vec{d}_A; m_B, \vec{d}_B)\right]\right) -1 = D - 1$. Proper faces of maximal dimension are called facets. A Bell inequality $B_{G_i} \cdot | P \rangle \leq \omega_c(G_i)$ thus supports a facet of the classical polytope if and only if $D$ affinely independent boxes of $\textbf{C}\left[\textbf{B}(2; m_A, \vec{d}_A; m_B, \vec{d}_B)\right]$ satisfy it with equality. 

Finding the quantum violation of a Bell inequality is also a well-known hard problem. In the special instance of two-party correlation Bell inequalities, also known as XOR games, the quantum value can be directly determined by means of a semi-definite program, as shown by Tsirelson \cite{Tsirelson}. For more general two-party Bell inequalities, where the parties observe more than two measurement outcomes, or where the inequality also involves marginal probabilities observed by either party, finding the quantum violation is not as easy. In \cite{NPA}, a hierarchy of semi-definite programs was formulated for optimization with non-commuting variables, and this NPA hierarchy is ubiquitously employed to efficiently determine upper bounds to the quantum violation for general Bell inequalities. The hierarchy was also shown to converge to a set $\textbf{Q}^{pr}\left[\textbf{B}(2; m_A, \vec{d}_A; m_B, \vec{d}_B)\right]$, which is the set consisting of boxes $P$ where each component $P_{O_A, O_B | I_A, I_B}(o_A, o_B | i_A, i_B)$ is obtained as:
\begin{equation}
P_{O_A, O_B | I_A, I_B}(o_A, o_B | i_A, i_B) = \text{Tr}\left[ \rho \left(E^{A}_{i_A, o_A}  E^{B}_{i_B, o_B} \right) \right],
\end{equation}
with $\left[E^{A}_{i_A, o_A}, E^{B}_{i_B, o_B} \right] = 0$ for all $i_A, o_A, i_B, o_B$. The above differs from Eq.(\ref{eq:qcorr-def}) in that the strict requirement of tensor product structure is replaced with the requirement of only commutation between different parties' measurements. It is clear that $\textbf{Q}\left[\textbf{B}(2; m_A, \vec{d}_A; m_B, \vec{d}_B)\right] \subseteq \textbf{Q}^{pr}\left[\textbf{B}(2; m_A, \vec{d}_A; m_B, \vec{d}_B)\right]$. 

In the NPA hierarchy, one considers sets consisting of sequences of product projection operators $S_1 = \{ \mathds{1} \} \cup \{E^{A}_{i_A, o_A} \} \cup \{E^{B}_{i_B, o_B}\}$, $S_2 = S_1 \cup \{E^{A}_{i_A, o_A} E^{B}_{i_B, o_B}\}$, etc. The convex sets corresponding to different levels of this hierarchy $\textbf{Q}_{l}\left[\textbf{B}(2; m_A, \vec{d}_A; m_B, \vec{d}_B)\right]$ are constructed by testing for the existence of a certificate $\Gamma^{l}$ associated to the set of operators $S_l$ by means of a semi-definite program. This certificate $\Gamma^{l}$ corresponding to level $l$ of the NPA hierarchy is a $\left|S_l \right| \times |S_l|$ matrix whose rows and columns are indexed by the operators in the set $S_l$. The certificate $\Gamma^{l}$ is required to be a complex Hermitian positive semi-definite matrix satisfying the following constraints on its entries: (i) $\Gamma^{l}_{\mathds{1}, \mathds{1}} = 1$, and (ii) $\Gamma^{l}_{Q, R} = \Gamma^{l}_{S, T}$ if $Q^{\dagger} R = S^{\dagger} T$. The latter condition in particular imposes that $\Gamma^{l}_{\mathds{1}, E^{A}_{i_A, o_A} E^{B}_{i_B, o_B}} = \Gamma^{l}_{E^{A}_{i_A, o_A}, E^{B}_{i_B, o_B}} = \Gamma^{l}_{E^{A}_{i_A, o_A} E^{B}_{i_B, o_B}, E^{A}_{i_A, o_A} E^{B}_{i_B, o_B}} = P_{O_A, O_B | I_A, I_B}(o_A, o_B | i_A, i_B)$. 

One of the levels of the NPA hierarchy denoted  $\textbf{Q}_{1+AB}\left[\textbf{B}(2; m_A, \vec{d}_A; m_B, \vec{d}_B)\right]$ or $\tilde{\textbf{Q}}\left[\textbf{B}(2; m_A, \vec{d}_A; m_B, \vec{d}_B)\right]$ has been highlighted as being the \textit{Almost Quantum set} \cite{AQ}. This set corresponds to an intermediate level of the hierarchy and is associated to the set of operators $\tilde{S} = \{\mathds{1} \} \cup \{E^{A}_{i_A, o_A} E^{B}_{i_B, o_B} \}$, where the latter set includes measurement operators for every $i_A, o_A, i_B, o_B$. Interestingly, this set has been proven to satisfy many of the information-theoretic principles designed to pick out quantum theory from among all no-signalling theories, such as the Local Orthogonality Principle, No advantage in Non-local computation etc. \cite{LO, NLC}. Moreover, a number of Bell inequalities achieve their optimal quantum violations already at this level, including the aforementioned correlation Bell inequalities.

\subsection{Bell inequalities with no quantum violation}

In identifying Bell inequalities for which no quantum violation exists, facet Bell inequalities play a crucial role. On the one hand, finding a facet Bell inequality with no quantum violation implies finding the largest dimensional face of the set of quantum correlations which one can describe analytically. On the other hand, if we relax the facet requirement, one can readily construct many Bell inequalities with no quantum violation by suitably tilting known facet Bell inequalities (that do admit quantum violation). 

(i) For instance, consider the well-studied CHSH Bell scenario $\textbf{B}\left((2,2),(2,2) \right)$, where Alice measures one of two binary observables $A_1, A_2$ and Bob similarly measures binary observables $B_1, B_2$. The classical polytope in this scenario is a well-characterized $8$-dimensional polytope with the only non-trivial facet (the trivial facets are the non-negativity constraints $P_{O_A, O_B | I_A, I_B}(o_A, o_B | i_A, i_B) \geq 0$) known to be the CHSH inequality (up to local relabelings of inputs and outputs and an exchange of parties) given as:
\begin{equation}
\langle A_1 B_1 \rangle + \langle A_1 B_2 \rangle + \langle A_2 B_1 \rangle - \langle A_2 B_2 \rangle \leq 2,
\end{equation}
where as usual the correlator is $\langle A_{i_A} B_{i_B} \rangle = \sum_{k=0,1} (-1)^k P_{O_A, O_B | i_A , i_B}(o_A \oplus o_B = k | i_A, i_B)$ for $i_A,i_B = 1,2$. Tilting the above facet inequality by choosing coefficients $\alpha_{11}, \alpha_{12}, \alpha_{21}, \alpha_{22} > 0$ normalized as $\alpha_{11} + \alpha_{12} + \alpha_{21} + \alpha_{22} = 1$, one gets the following class of inequalities 
\begin{eqnarray}
\label{Eq:chsh-no-adv}
&&\alpha_{11} \langle A_1 B_1 \rangle + \alpha_{12} \langle A_1 B_2 \rangle + \alpha_{21} \langle A_2 B_1 \rangle - \alpha_{22} \langle A_2 B_2 \rangle \nonumber \\ &&\qquad  \leq 1 - 2 \min \{\alpha_{11}, \alpha_{12}, \alpha_{21}, \alpha_{22} \}.
\end{eqnarray}
Using the Tsirelson solution for the quantum value of correlation Bell inequalities with binary outcomes, a simple characterization for the XOR games with no quantum advantage was obtained in \cite{RKMH14}. We can use the characterization to show that (non-facet) Bell inequalities of the form in (\ref{Eq:chsh-no-adv}) do not admit quantum violation when the following condition is satisfied by the coefficients (in the case when $\alpha_{22} < \alpha_{11}, \alpha_{12}, \alpha_{21}$) \cite{RQSMA17}:
\begin{eqnarray}
\label{eq:chsh-noad-cond}
&&\left(\alpha_{12} \alpha_{21} + \alpha_{11} \alpha_{22} \right)^2 \nonumber \\
&&\leq \left(\alpha_{11} + \alpha_{12} \right)\left(\alpha_{11} + \alpha_{21} \right)\left(\alpha_{12} - \alpha_{22} \right)\left(\alpha_{21} - \alpha_{22} \right).
\end{eqnarray}
An analogous condition holds when one of the other coefficients is the minimum as well. As an example satisfying the above condition, one may take $\{\alpha_{11}, \alpha_{12}, \alpha_{21}, \alpha_{22} \} = \left\{\frac{9}{16}, \frac{1}{4}, \frac{1}{8}, \frac{1}{16} \right\}$. 

(ii) A second important consideration in finding Bell inequalities with no quantum violation is a recent breakthrough result \cite{ECW20} showing that any two-player XOR game, for which the corresponding Bell inequality is tight, has a quantum advantage. Their result, automatically rules out inequalities such as (\ref{Eq:chsh-no-adv}) under condition (\ref{eq:chsh-noad-cond}) and the XOR games with no quantum advantage derived in \cite{NLC, RKMH14} from being facet Bell inequalities. However, binary-outcome correlation Bell inequalities only form a small subset of possible two-party Bell inequalities, since they restrict to the case $\d_{A, i_A} = d_{B, i_B} = 2$ for all $i_A, i_B$ and furthermore to the case that the inequality only consider terms involving the correlators $\langle A_{i_A} B_{i_B} \rangle$. Indeed, correlation Bell inequalities directly generalize to Bell scenarios where the number of outcomes for each party is $d > 2$ leading to Bell inequalities of the type
\begin{eqnarray}
\label{eq:corr-d-dim}
&&\sum_{i_A = 1}^{m_A} \sum_{i_B = 1}^{m_B} \sum_{o_A, o_B = 1}^{d} q\left(i_A, i_B\right) \nonumber \\
&&\quad P_{O_A, O_B | i_A , i_B}\left(o_A + o_B \; \text{mod} \; d = f(i_A, i_B) | i_A, i_B \right) \leq \beta_c, \nonumber \\
\end{eqnarray}
for some function $f$ from the inputs $(i_A, i_B)$ to a value in $\{1, \dots, d\}$. Let the root of unity be $\zeta = \exp \left(2 \pi i / d\right)$, and define $d-1$ (game) matrices of dimension $m_A \times m_B$ as
\begin{eqnarray}
\Phi_k := \sum_{i_A = 1}^{m_A} \sum_{i_B = 1}^{m_B} q\left(i_A, i_B \right) \zeta^{k f(i_A, i_B)} |i_A \rangle \langle i_B |,
\end{eqnarray}
for $k = 1, \dots, d-1$. Then, a sufficient condition for Bell inequalities of the form (\ref{eq:corr-d-dim}) to have no quantum violation was shown by us in \cite{RQSMA17}. Namely, if the maximum left and right singular vectors $|u_{1} \rangle$ and $|v_1 \rangle$ of $\Phi_1$ are composed entirely of roots of unity entries alone (arbitrary integral powers of $\zeta$), and simultaneously if the maximum left and right singular vectors $|u_k \rangle$ and $|v_k \rangle$ of $\Phi_k$ are obtainable from $|u_1 \rangle$ and $|v_1 \rangle$ by the substitution $\zeta \rightarrow \zeta^k$, then the corresponding inequality (\ref{eq:corr-d-dim}) admits no quantum violation. As an example consider the inequality corresponding to the game matrix
\[\Phi_1 = \frac{1}{24}
\begin{bmatrix}
    i  & 2 & -2  & i \\
    2 & i & i & -2 \\
    -2 & i & i & 2 \\
    i  & -2 & 2 & i
\end{bmatrix}
\]
i.e. with $f(1,1) = f(2,2) = f(3,3) = f(4,4) = f(1,4) = f(2,3) = f(3,2) = f(4,1) = 1$, $f(1,2) = f(2,1) = f(3,4) = f(4,3) = 4$, $f(1,3) = f(2,4) = f(3,1) = f(4,2) = 2$ and similarly $q(1,2) = q(2,1) = q(3,4) = q(4,3) = q(1,3) = q(2,4) = q(3,1)=q(4,2) = \frac{1}{12}$ and the remaining eight probabilities all equal to $\frac{1}{24}$. The corresponding inequality has classical value $\beta_c = \frac{3}{4}$ and an optimization to level $\textbf{Q}_1\left[\textbf{B}(2; 4, (4,4,4,4); 4, (4,4,4,4))\right]$ shows that the quantum value is also equal to $\beta_q = \frac{3}{4}$. This is reflected in the maximum singular vectors of $\Phi_1$ being composed of powers of $i = \exp\left(2 \pi i / 4\right)$ only, and the corresponding condition being satisfied by the matrices $\Phi_2, \Phi_3$ as well. The fact that the sufficient condition is inherited from the level $\textbf{Q}_1\left[\textbf{B}(2; m_A, \vec{d}_A; m_B, \vec{d}_B)\right]$ in general \cite{Tsirelson} implies, by our central result, that none of the corresponding Bell inequalities with no quantum violation define facets of the Bell polytope. 

(iii) Furthermore, the fact that XOR games obey a perfect parallel repetition theorem \cite{CSUU07}, implies that from a given binary-outcome correlation Bell inequality with no quantum violation, one can construct several Bell inequalities in higher-dimensional multi-outcome Bell scenarios that also do not allow for quantum violation. Indeed, any parallel repetition of the non-local computation game from \cite{NLC} yields examples of $2^k$-output games without quantum advantage.   

(iv) Even considering Bell inequalities with marginal terms, it is possible to construct inequalities with no quantum violation. We give an illustrative example here in the simple $\textbf{B}\left((2,2),(2,2) \right)$ scenario, more involved scenarios require a careful construction using the NPA hierarchy. Consider the inequality parametrized by real $0 \leq \alpha \leq 1$ and given as
\begin{widetext}
\begin{eqnarray}
\left(P_{O_A, O_B | i_A , i_B}(0,0 | 0,0) + \alpha P_{O_A, O_B | i_A , i_B}(1,1 | 0,0) \right)-
P_{O_A, O_B | i_A , i_B}(0,1 | 0,1) - P_{O_A, O_B | i_A , i_B}(1,0 | 1,0) - P_{O_A, O_B | i_A , i_B}(0,0 | 1,1) \leq \alpha, \nonumber \\
\end{eqnarray}
\end{widetext}
where the classical maximum of $\alpha$ is readily obtained by direct inspection over local deterministic strategies. A well-known result using Jordan's lemma \cite{Jordan} states that the quantum maximum of inequalities in this Bell scenario is obtainable by performing projective measurements on two-qubit states. On the other hand, the no-signalling violation of the inequality is achieved by a Popescu-Rohrlich box \cite{PR94} which assigns value $1/2$ to the two terms in the bracket and $0$ to the remaining probabilities, to give the maximum no-signaling value of $\frac{1+\alpha}{2}$. A direct optimization over two-qubit states reveals that the inequality has the quantum value $\beta_q = \alpha < \frac{1+\alpha}{2}$ in the parameter range $0.867 \leq \alpha < 1$. Three affinely independent local deterministic strategies saturate the inequality, showing that the inequality defines a two-dimensional face of the set of quantum correlations (this is one less than the bound of $m_A + \frac{1}{2} m_A \left(m_A - 1 \right)$ derived in \cite{ECW20}. These are explicitly given as follows: (i) Alice outputs $(1,1)$ for her two inputs $i_A = 1,2$, Bob outputs $(1,0)$ for $i_B = 1,2$, (ii) Alice outputs $(1,1)$, Bob outputs $(1,1)$, (iii) Alice outputs $(1,0)$, Bob outputs $(1,1)$.  

\subsection{Facet Bell inequalities are violated in Almost Quantum theory}

Cabello, Severini and Winter discovered a relationship between Bell scenarios (that also extends to more general contextuality scenarios) and Graph theory \cite{CSW1, CSW2}. For a given two-party Bell scenario $\textbf{B}(\vec{d}_A, \vec{d}_B)$, one constructs an orthogonality graph $G_{\textbf{B}(\vec{d}_A, \vec{d}_B)}$ as follows. Each input-output combination $(o_A, o_B | i_A, i_B)$ corresponds to a distinct vertex $v_{(o_A, o_B | i_A, i_B)}$ of the graph, and two such vertices are connected by an edge if the corresponding events are locally orthogonal, where local orthogonality is the condition that distinct outcomes are obtained for the same local input. In other words, we have
\begin{widetext}
\begin{equation}
v_{(o_A, o_B | i_A, i_B)} \sim v_{(o'_A, o'_B | i'_A, i'_B)} \Leftrightarrow \left(i_A = i'_A \wedge o_A \neq o'_A\right) \lor \left(i_B = i'_B \wedge o_B \neq o'_B\right).
\end{equation} 
\end{widetext}
Equivalently, we may consider that each product measurement operator $E^{A}_{i_A, o_A} E^{B}_{i_B, o_B}$ corresponds to a vertex in the graph  $G_{\textbf{B}(\vec{d}_A, \vec{d}_B)}$ with vertices connected by an edge if $i_A = i'_A$ and $E^{A}_{i_A, o_A} E^{A}_{i'_A, o'_A} = 0$ or $i_B = i'_B$ and $E^{B}_{i_B, o_B} E^{B}_{i'_B, o'_B} = 0$. The number of vertices in the graph is  $\Big| V\left(G_{\textbf{B}(\vec{d}_A, \vec{d}_B)} \right) \Big| = n$.  

Furthermore, given a graph $G$ with vertex set $V(G)$ and edge set $E(G)$ one can also find a set of unit vectors obeying the above orthogonality conditions, called an orthonormal representation of the graph. Formally, an orthonormal representation of graph $G$ is a set $\{ | u_i \rangle \in \mathbb{R}^N: i \in V(G) \}$ where $N$ is some arbitrary dimension, $\| |u_i \rangle \| = 1$ for all $i \in V(G)$ and $\langle u_i | u_j \rangle = 0$ for $(i, j) \in E(G)$. It should be noted that in the graph-theoretic literature, the Lov\'{a}sz orthogonal representation is also defined in a complementary manner, with non-adjacent vertices being assigned orthogonal vectors. For a given graph $G$, the Lov\'{a}sz theta-body $\text{TH}(G)$ (sometimes also called the Grötschel-Lovász-Schrijver theta-body)  is a convex set introduced \cite{Lovasz-1, Lovasz-2, Schrijver} as a semi-definite programming relaxation to the hard graph-theoretic problem of finding a maximum weight stable set of the graph (a stable set is a set of mutually non-adjacent vertices). The theta set is defined as follows:
\begin{definition}
For a graph $G = \left(V(G), E(G) \right)$, define the convex set $TH(G)$ as
\begin{widetext}
\begin{equation}
\text{TH}(G) := \Bigg\{| \mathcal{P} \rangle = \left( |\langle \psi | u_i \rangle|^2 : i \in V(G) \right) \in \mathbb{R}_{+}^{V(G)} \Bigg| \begin{array}{l}
    \| | \psi \rangle \| = \| | u_i \rangle \| = 1, \\
    \{| u_i \rangle \} \; \text{is an orthonormal representation of G}
  \end{array}\Bigg\}
\end{equation}
\end{widetext}
\end{definition}

The similarity between the set $TH(G)$ and the set $\textbf{Q}_{1+AB}\left[\textbf{B}(2; m_A, \vec{d}_A; m_B, \vec{d}_B)\right]$ has been noted in the literature, here we give a self-contained proof that is more suited towards establishing our main result. Firstly, as shown in \cite{LO}, the normalization and no-signalling constraints on a box can be rewritten in terms of maximum clique equalities in the orthogonality graph, where a clique inequality is an inequality of the form 
\begin{equation}
\sum_{v_{(o_A, o_B | i_A, i_B)} \in c} P_{O_A, O_B | I_A, I_B}(o_A, o_B | i_A, i_B) \leq 1,
\end{equation}
for some clique $c$ in the graph. Here, a clique denotes a set of mutually adjacent vertices. Now, since by definition, each normalization constraint only considers events corresponding to different outcomes for the same measurement setting, it is clear that the normalization constraint corresponds to a clique inequality that is saturated. To see that the no-signaling condition also corresponds to such a constraint, note that using the normalization constraint, the no-signaling conditions can be rewritten in the form
\begin{widetext}
\begin{eqnarray}
\sum_{o_A = 1}^{d_{A, i_A}} P_{O_A, O_B | I_A, I_B}(o_A, o_B | i_A, i_B) + \sum_{\substack{o'_B = 1 \\ o'_B \neq o_B}}^{d_{B, i_B}} \sum_{o'_A = 1}^{d_{A, i'_A}} P_{O_A, O_B | I_A, I_B}(o'_A, o'_B | i'_A, i_B) &=& 1, \quad \text{for all} \; i_A, i'_A, o_B, i_B \nonumber \\
\sum_{o_B = 1}^{d_{B, i_B}} P_{O_A, O_B | I_A, I_B}(o_A, o_B | i_A, i_B) + \sum_{\substack{o'_A = 1 \\ o'_A \neq o_A}}^{d_{A, i_A}} \sum_{o'_B = 1}^{d_{B, i'_B}} P_{O_A, O_B | I_A, I_B}(o'_A, o'_B | i_A, i'_B) &=& 1, \quad \text{for all} \; i_B, i'_B, o_A, i_A.
\end{eqnarray}
\end{widetext}
Each no-signaling condition expressed in the above form considers events that are locally orthogonal, and thus corresponds to a saturated clique inequality. Furthermore, the normalization and no-signalling conditions correspond to \textit{maximum} clique inequalities, i.e., no other measurement event $(o_A, o_B | i_A, i_B)$ exists that is locally orthogonal to every event in these equations. Interestingly, it was shown in \cite{LO} that in any two-party Bell scenario $\textbf{B}(\vec{d}_A, \vec{d}_B)$, the normalization and no-signalling conditions encompass all the maximum clique inequalities, i.e., every maximum clique inequality corresponds to a normalization or a no-signalling constraint. On the other hand, when one considers Bell scenarios involving three or more parties, other maximum clique inequalities exist, and these are the constraints identified by the Local Orthogonality principle. 

In a formal sense, $\textbf{Q}_{1+AB}\left[\textbf{B}(2; m_A, \vec{d}_A; m_B, \vec{d}_B)\right]$ is equivalent to the set $TH(G)$ defined for an appropriate orthogonality graph $G_{\textbf{B}(\vec{d}_A, \vec{d}_B)}$, with the additional constraint that the maximum clique inequalities corresponding to the normalization and the no-signalling conditions be set to equalities. In other words, define $\mathcal{C}_{n, ns}$ as the set of maximum cliques in the orthgonality graph $G_{\textbf{B}(\vec{d}_A, \vec{d}_B)}$ that correspond to the normalization and no-signalling constraints in the Bell scenario. Define the convex set $\text{TH}_{n, ns}\left(G_{\textbf{B}(\vec{d}_A, \vec{d}_B)} \right)$ as
\begin{widetext}
\begin{equation}
\text{TH}_{n, ns}\left(G_{\textbf{B}(\vec{d}_A, \vec{d}_B)} \right) := \Bigg\{ | \mathcal{P} \rangle  = \left( |\langle \psi | u_i \rangle|^2 : i \in V\left(G_{\textbf{B}(\vec{d}_A, \vec{d}_B)}\right) \right) \in \mathbb{R}_{+}^n \Bigg| \begin{array}{l}
    \| | \psi \rangle \| = \| | u_i \rangle \| = 1 \; \forall i  \\
    \{| u_i \rangle \} \; \text{is an orth. repn. of} \; G_{\textbf{B}(\vec{d}_A, \vec{d}_B)}, \\
    \sum_{i \in c} |\langle \psi | u_i \rangle|^2 = 1, \; \; \; \; \text{for all} \; c \in \mathcal{C}_{n, ns}
  \end{array}\Bigg\}
\end{equation}
\end{widetext}
The set $\textbf{Q}_{1+AB}\left[\textbf{B}(2; m_A, \vec{d}_A; m_B, \vec{d}_B)\right]$ is then equivalent to $\text{TH}_{n, ns}\left(G_{\textbf{B}(\vec{d}_A, \vec{d}_B)} \right)$:
\begin{thm}[see \cite{Fritz2, our, CSW1}]
For any two-party Bell scenario $\textbf{B}(\vec{d}_A, \vec{d}_B)$, it holds that $\textbf{Q}_{1+AB}\left[\textbf{B}(\vec{d}_A, \vec{d}_B)\right] = \text{TH}_{n, ns}\left(G_{\textbf{B}(\vec{d}_A, \vec{d}_B)} \right)$. 
\end{thm}
At this point, it is important to note the dimension mismatch between the sets $\textbf{Q}_{1+AB}\left[\textbf{B}(2; m_A, \vec{d}_A; m_B, \vec{d}_B)\right]$ and $TH(G)$. Namely, while $TH(G)$ is a full-dimensional convex set (of dimension $n$), $\textbf{Q}_{1+AB}\left[\textbf{B}(2; m_A, \vec{d}_A; m_B, \vec{d}_B)\right]$ is of much smaller dimension (being of dimension $D$). Therefore, one may wonder whether any statements about the facets of $TH(G)$ hold true for the smaller dimensional set, since facets of $\textbf{Q}_{1+AB}\left[\textbf{B}(2; m_A, \vec{d}_A; m_B, \vec{d}_B)\right]$ would be faces of much smaller dimension in $TH(G)$. Nevertheless, we use techniques used in the study of the facets of $TH(G)$ to show the following statement about the facets of the Almost Quantum set. 


\begin{thm}
\label{thm:main}
Every two-party facet Bell inequality, irrespective of the number of inputs and outputs for each party, admits a violation in almost quantum theory. In other words, for a facet Bell inequality of the form
 \begin{equation}
 \sum_{o_A,o_B,i_A,i_B}  q(i_A,i_B) V(o_A,o_B,i_A,i_B) P(o_A,o_B|i_A,i_B) \leq \omega_c 
 \end{equation}
where $\omega_c$ denotes the classical value of the inequality, the almost quantum value $\omega_{\tilde{q}}$ is strictly larger than $\omega_c$, i.e., $\omega_{\tilde{q}} > \omega_c$. 
\end{thm}
 
 \begin{proof}
The proof follows analogously to that of an analogous claim made for the general Lov\'{a}sz theta set $\text{TH}(G)$. It is noteworthy that the set $\text{TH}(G)$ has been characterized in multiple ways in the literature. We begin with a complementary characterization of the set $\text{TH}_{n, ns}\left(G_{\textbf{B}(\vec{d}_A, \vec{d}_B)} \right)$ inherited from a characterization of $\text{TH}(G)$ \cite{Knuth, Schrijver2} that is particularly suited to our problem. 
\begin{widetext}
\begin{equation}
\label{eq:TH-nns-def}
\text{TH}_{n,ns}\left(G_{\textbf{B}(\vec{d}_A, \vec{d}_B)} \right) = \left\{  | \mathcal{P} \rangle \in \mathbb{R}_{+}^n  \Bigg\vert \;\; \begin{array}{l}
\sum_{i \in V\left(G_{\textbf{B}(\vec{d}_A, \vec{d}_B)}\right)} |\langle \phi | w_i \rangle|^2 | \mathcal{P}\rangle_i \leq 1, \\
\{ | w_i \rangle\} \; \text{is an orth. repn. of} \; \overline{G}_{\textbf{B}(\vec{d}_A, \vec{d}_B)}, \\
\| | \phi \rangle \| = \| | w_i \rangle \| = 1 \; \forall i \\
\sum_{i \in c} | \mathcal{P} \rangle_i = 1, \; \; \; \; \text{for all} \; c \in \mathcal{C}_{n, ns}
\end{array}\right\}
\end{equation}
\end{widetext}
Here $\overline{G}$ denotes the graph complement of $G$, i.e., the graph with the same vertex set as $G$ and the complementary edge set ($u \sim v$ in $\overline{G}$ $\Leftrightarrow$ $u \not\sim v$ in $G$). This representation of $\text{TH}_{n,ns}\left(G_{\textbf{B}(\vec{d}_A, \vec{d}_B)} \right) = \textbf{Q}_{1+AB}\left[\textbf{B}(2; m_A, \vec{d}_A; m_B, \vec{d}_B)\right]$ is useful since it characterizes the facets of the set, in particular every facet is of the form $\sum_{i \in V\left(G_{\textbf{B}(\vec{d}_A, \vec{d}_B)}\right)} |\langle \phi | w_i \rangle|^2 | \mathcal{P}\rangle_i = 1$ for some unit vector $|\phi \rangle \in \mathbb{R}^N$ and orthonormal representation $\{ |w_i \rangle \in \mathbb{R}^N \}$ of $\overline{G}$. It is also worth noting that the normalization and no-signalling constraints $\sum_{i \in c} | \mathcal{P} \rangle_i = 1$ also fall in this category, if we choose $|w_i \rangle = | \phi \rangle$ for every vertex $i \in c$ and $|w_i \rangle = | \phi \rangle^{\perp}$, for some arbitrary unit vector $| \phi \rangle$ and an orthogonal unit vector $| \phi \rangle^{\perp} \perp | \phi \rangle$. 

Let $F = \{ | P \rangle \big| \sum_{i  \in V\left(G_{\textbf{B}(\vec{d}_A, \vec{d}_B)}\right)} |\langle \phi | w_i \rangle|^2 | \mathcal{P}\rangle_i = 1\}$ be a facet of $\textbf{Q}_{1+AB}\left[\textbf{B}(2; m_A, \vec{d}_A; m_B, \vec{d}_B)\right]$. Let $| \mathcal{P}^* \rangle \in \text{int}(F)$. We have the following
\begin{eqnarray}
\label{eq:facet-der-1}
\sum_{i  \in V\left(G_{\textbf{B}(\vec{d}_A, \vec{d}_B)}\right)} |\langle \phi | w_i \rangle|^2 | \mathcal{P}^* \rangle_i &\leq& 1 \nonumber \\
\implies \sum_{i  \in V\left(G_{\textbf{B}(\vec{d}_A, \vec{d}_B)}\right)} |\langle \phi | w_i \rangle|^2 | \mathcal{P}^* \rangle_i &\leq& \langle \phi | \phi \rangle \nonumber \\
\implies \langle \phi | \left( \sum_{i  \in V\left(G_{\textbf{B}(\vec{d}_A, \vec{d}_B)}\right)}  | \mathcal{P}^* \rangle_i | w_i \rangle \langle w_i | \right) | \phi \rangle &\leq& 1
\end{eqnarray}
Saturation of the above inequality implies that $| \phi \rangle$ is a maximum eigenvector of $\left( \sum_{i  \in V\left(G_{\textbf{B}(\vec{d}_A, \vec{d}_B)}\right)}  | \mathcal{P}^* \rangle_i | w_i \rangle \langle w_i | \right)$ corresponding to eigenvalue $1$.  In other words
\begin{widetext}
\begin{eqnarray}
\left( \sum_{i  \in V\left(G_{\textbf{B}(\vec{d}_A, \vec{d}_B)}\right)}  | \mathcal{P}^* \rangle_i | w_i \rangle \langle w_i | \right) | \phi \rangle &=& | \phi \rangle \nonumber \\
\left( \sum_{i  \in V\left(G_{\textbf{B}(\vec{d}_A, \vec{d}_B)}\right)}  | \mathcal{P}^* \rangle_i \langle w_i | \phi \rangle \right)  | w_i \rangle_{j} &=& | \phi \rangle_{j} \; \; \text{for} \; j = 1, \dots, N \nonumber \\
\left( \sum_{i  \in V\left(G_{\textbf{B}(\vec{d}_A, \vec{d}_B)}\right)}  | \mathcal{P}^* \rangle_i \langle w_i | \phi \rangle \right)  | w_i \rangle_{j} &=& \left( \sum_{i  \in V\left(G_{\textbf{B}(\vec{d}_A, \vec{d}_B)}\right)} | \mathcal{P}^* \rangle_i |\langle w_i | \phi \rangle|^2   \right) | \phi \rangle_{j} \; \; \text{for} \; j = 1, \dots, N,
\end{eqnarray}
\end{widetext}
where in obtaining the last equation we have used the first inequality of (\ref{eq:facet-der-1}). Now as $F$ is a facet, it is not a convex combination of other facet-defining inequalities in Eq.(\ref{eq:TH-nns-def}), so that the above equality implies
\begin{eqnarray}
\langle w_i | \phi \rangle  | w_i \rangle &=& |\langle w_i | \phi \rangle|^2 |\phi \rangle \; \; \forall i \in V\left(G_{\textbf{B}(\vec{d}_A, \vec{d}_B)}\right).
\end{eqnarray}
This gives that for each $i \in V\left(G_{\textbf{B}(\vec{d}_A, \vec{d}_B)}\right)$, either $\langle w_i | \phi \rangle = 0$ or 
\begin{eqnarray}
| w_i \rangle &=& \langle w_i | \phi \rangle |\phi \rangle \nonumber \\
\implies |w_i \rangle &=& \pm |\phi \rangle. 
\end{eqnarray}
Without loss of generality we may take $|w_i \rangle = |\phi \rangle$. Therefore, for every $i \in V\left(G_{\textbf{B}(\vec{d}_A, \vec{d}_B)}\right)$, either $\langle w_i | \phi \rangle = 0$ or $|w_i \rangle = | \phi \rangle$. Defining the set $\mathcal{I}$ as $\mathcal{I} := \{i \in V\left(G_{\textbf{B}(\vec{d}_A, \vec{d}_B)}\right) \Big| \; \;  |w_i \rangle = |\phi \rangle \}$, we obtain that $\{ | w_i \rangle \}$ is an orthonormal representation of $\overline{G}$ where $| w_i \rangle$ takes value $| \phi \rangle$ for every $i \in \mathcal{I}$, while $| w_i \rangle$ belongs to the a subspace of $\mathbb{R}^N$ that is orthogonal to $| \phi \rangle$ when $i \notin \mathcal{I}$. This therefore implies that all the vertices in $\mathcal{I}$ are mutually non-adjacent in $\overline{G}$, i.e., that $\mathcal{I}$ is a stable set of $\overline{G}$ or in other words $\mathcal{I}$ is a clique of $G$. The inequality $\sum_{i  \in V\left(G_{\textbf{B}(\vec{d}_A, \vec{d}_B)}\right)} |\langle \phi | w_i \rangle|^2 | \mathcal{P}\rangle_i = 1$ supporting facet $F$ of $\textbf{Q}_{1+AB}\left[\textbf{B}(2; m_A, \vec{d}_A; m_B, \vec{d}_B)\right]$ is thus a clique inequality of the form $\sum_{i \in \mathcal{I}} | \mathcal{P} \rangle_i = 1$. Now from \cite{LO}, we know that in any two-party Bell scenario, the clique inequalities are exhausted by the no-signalling and normalization constraints. Therefore, no other proper facet-defining Bell inequality exists that is also a facet of $\textbf{Q}_{1+AB}\left[\textbf{B}(2; m_A, \vec{d}_A; m_B, \vec{d}_B)\right]$, i.e., every two-party facet Bell inequality is violated in Almost Quantum theory.  
\end{proof}
 
It is worth noting that an alternative route to deriving the result in Thm. (\ref{thm:main}) is to parametrize box $| \mathcal{P} \rangle$ in terms of $D$ parameters, being probabilities $P_{O_A, O_B | I_A, I_B}(o_A, o_B | i_A, i_B)$, from which other probabilities that define the box can be obtained via normalization and no-signalling conditions, as done for example in \cite{NPA}. The result on facets of $TH(G)$ from \cite{Knuth, Schrijver2} applied to the theta-set of the orthogonality graph of this subset of events in the Bell experiment, can then be used to derive Thm. (\ref{thm:main}).  
 

\begin{corollary}
Any two-party facet Bell inequality, irrespective of the number of inputs and outputs for each party, for which the quantum value is achieved at level $1+AB$ of the NPA hierarchy, admits a violation in quantum theory. In particular, two-party XOR games that define facet Bell inequalities always admit a quantum advantage. 
\end{corollary}
\begin{proof}
Two-party XOR games form a class of Bell inequalities for which the quantum value is achieved at level $1+AB$, in fact already at level $1$ of the NPA hierarchy, by the results of Tsirelson \cite{Tsirelson}. Therefore, binary-outcome correlation Bell inequalities that do not admit quantum violation, do not define facets of the classical Bell polytope. 
\end{proof}

This Corollary neatly recovers the result by Escol\'{a} et al. Besides, as we have seen it also recovers a central result of \cite{RQSMA17}, namely that $d$-outcome non-local computation games do not define facets of the Bell polytope. Finally, other unique games with no quantum advantage considered in \cite{RQSMA17} are also shown to not correspond to facet-defining Bell inequalities. An important question remains, whether the methods can be extended to other levels of the NPA hierarchy, to identify whether any two-party facet Bell inequality also defines a facet of the set of quantum correlations. Such a facet, if it exists, would provide a fundamental information-theoretic principle, to identify why nature chose Quantum theory over Almost Quantum theory \cite{AQ}. 

\subsection{All two-party facet-defining inequalities of the Bell correlation polytope are violated in Quantum theory}
Avis et al. \cite{AII06} posed the question whether there are any facets of the binary-outcome correlation polytope (the classical polytope of two-party binary-outcome correlations $\langle A_{i_A} B_{i_B} \rangle \in \{+1,-1\}$ , excluding the local marginal terms) that are not violated in Quantum theory. This question was recently answered in the negative by Escol\'{a} et al. \cite{ECW20}, making use of the simple characterization of XOR games with no quantum advantage given in \cite{RKMH14}. Here, we provide an alternative proof, making use of a connection between the correlation polytope and the Cut polytope of graph theory \cite{AII06}. In particular, this connection links the set of binary-outcome quantum correlations and the well-studied elliptope in graph theory, this latter body being the semi-definite programming relaxation of the Cut polytope. 

First, we introduce the cut polytope of complete graph. The graph is denoted by $\Gamma_t$, has $t$ vertices, and has an edge between each pair of vertices. A cut $S$ is an assignment of $\{0,1\}$ to each vertex in the graph. The cut vector $\delta(S)$ for some cut $S$ is given by $\delta_{u,v}(S) = 1$ if vertices $u, v$ are assigned different values, and $0$ if the vertices are assigned the same values. The set of all convex combinations of cut vectors $\text{CUT}(\Gamma_t) = \{\sum_{S: cut} p_S \delta(S) | \sum_{S: cut} p_S = 1, \; p_S \geq 0\}$ is called the Cut polytope of the complete graph. The vectors of correlation functions which are possible in classical correlation experiments form the cut polytope $\text{Cut}\left(\Gamma_{m_A,m_B}\right)$ of the complete bipartite graph $\Gamma_{m_A,m_B}$. Tight correlation inequalities are exactly the facet-inducing inequalities of the polytope $Cut\left(\Gamma_{m_A,m_B} \right)$.

The semi-definite relaxation of the cut polytope of the complete graph $\Gamma_t$ is the elliptope $\mathcal{E}(\Gamma_t)$ also sometimes denoted as $\mathcal{E}_{t \times t}$. Formally, $\mathcal{E}_{t \times t}$ denotes the set of $t \times t$ correlation matrices (positive semidefinite matrices with diagonal entries equal to $1$)
\begin{equation}
\mathcal{E}_{t \times t} := \left\{ M \in \mathbb{R}^{t \times t} \Big| M \succeq 0, M_{i,i} = 1 \; \text{for all} \; i = 1, \dots, t \right\}.
\end{equation}
In general, the elliptope $\mathcal{E}(G)$ of a graph $G = (V, E)$ with  $|V |$ vertices is the convex body consisting of vectors $\vec{x} \in \mathbb{R}^{E}$ such that there exist a unit vector $|u_i \rangle \in \mathbb{R}^{|V|}$ for each vertex $i \in V$ satisfying $\vec{x}_{i,j} = \langle u_i | u_j \rangle$. In particular, the elliptope of the complete bipartite graph $\mathcal{E}(\Gamma_{m_A, m_B})$ is the set of vectors $\vec{x} \in \mathbb{R}^{E(\Gamma_{m_A, m_B})}$ satisfying the conditions of Tsirelson's theorem, so that $\mathcal{E}\left(\Gamma_{m_A, m_B} \right)$ is the set of bipartite binary-outcome quantum correlations \cite{AII06}. The dimensionalities of these sets is $\text{dim}\left(\mathcal{E}_{t \times t} \right) = {t \choose 2}$, and $\text{dim}\left(\mathcal{E}\left(\Gamma_{m_A, m_B} \right) \right) = m_A m_B$. Moreover, $\mathcal{E}\left(\Gamma_{m_A, m_B} \right)$ is a projection of $\mathcal{E}_{t \times t}$ onto the lower-dimensional space for $t = m_A + m_B$.  


We now show that every two-party facet-defining correlation Bell inequality, irrespective of the number of inputs and outputs for each party, admits a violation in quantum theory. In other words, for a facet-defining correlation Bell inequality of the form
 \begin{equation}
 \sum_{i_A, i_B} \alpha_{i_A, i_B} \langle A_{i_A} B_{i_B} \rangle \leq \beta_c 
 \end{equation}
where $\beta_c$ denotes the classical value of the inequality, the quantum value $\beta_{q}$ is strictly larger than $\beta_c$.

The proof comes from the discussion above, mapping the Cut polytope and the correlation Bell polytope, along with the corresponding mapping between the Elliptope and the set of two-party binary-outcome quantum correlations \cite{AII06}. Laurent and Poljak in \cite{LP96}, building upon the results of \cite{LT94} show that the largest dimension of a polyhedral face (formed by the convex hull of cut vectors) of $\mathcal{E}_{t \times t}$ is equal to the largest integer $d_t$ such that ${d_t+1 \choose 2} \leq t-1$, i.e., $d_t = \lfloor \frac{\sqrt{8t-7} -1}{2} \rfloor$. They further show that the largest dimension of any face of $\mathcal{E}_{t \times t}$ is ${t-1 \choose 2}$.  A facet-defining correlation Bell inequality is, by definition, of dimension $\text{dim}\left(\mathcal{E}\left(\Gamma_{m_A, m_B} \right) \right) - 1 = m_A m_B - 1$. For $t := m_A + m_B$, it is readily seen that for all values of $m_A, m_B \geq 2$, $m_A m_B - 1 \geq d_t$. The mapping also allows to derive novel Quantum Bell inequalities, which may be of interest in self-testing applications, as we show in the next section.  

\subsection{Quantum Bell Inequalities}

The characterization Eq.(\ref{eq:TH-nns-def}) of $\textbf{Q}_{1+AB}\left[\textbf{B}(\vec{d}_A, \vec{d}_B)\right] = \text{TH}_{n, ns}\left(G_{\textbf{B}(\vec{d}_A, \vec{d}_B)} \right)$ used in the proof of Theorem \ref{thm:main} is particularly useful in deriving novel (Almost) Quantum Bell inequalities, that provide candidate Quantum Bell inequalities for quantum self-testing applications \cite{Scarani}. In particular, Eq.(\ref{eq:TH-nns-def}) defines the boundary of the Almost Quantum set in terms of the linear inequalities
\begin{eqnarray}
\label{eq:AQ-boundary}
\sum_{i \in V\left(G_{\textbf{B}(\vec{d}_A, \vec{d}_B)}\right)} |\langle \phi | w_i \rangle|^2 | \mathcal{P}\rangle_i \leq 1, 
\end{eqnarray}
where $\{ | w_i \rangle \in \mathbb{R}^N\}$ is an orthonormal representation of the complement graph $\overline{G}_{\textbf{B}(\vec{d}_A, \vec{d}_B)}$ and $|\phi \rangle \in \mathbb{R}^N$ is an arbitrary dimensional unit vector. Boxes satisfying (\ref{eq:AQ-boundary}) with equality define a linear boundary, that may be investigated to check if any quantum boxes living on the boundary admit a self-testing quantum realization \cite{Scarani}. As an example, we investigate the CHSH Bell scenario $\textbf{B}\left((2,2), (2,2)\right)$, where $\Big| V\left(G_{\textbf{B}\left((2,2), (2,2)\right)}\right)\Big| = 16$. The boundary of the set of quantum correlations (excluding the local marginal terms) in this scenario was characterized by Tsirelson \cite{Tsirelson} to be
\begin{eqnarray}
\label{eq:QBell-chsh}
\sum_{(x,y) \neq (i,j)} \arcsin\left(\langle A_x B_y \rangle\right) - \arcsin\left(\langle A_i B_j \rangle\right) = \xi \pi,
\end{eqnarray}  
where $i,j \in \{1,2\}$, $\xi = \pm 1$. Quantum correlations that satisfy the above condition were shown to self-test the two-qubit singlet state in \cite{Scarani}. 

On the other hand, linear inequalities of the form (\ref{eq:AQ-boundary}) bound the entire $8$-dimensional Almost Quantum boundary including the marginal terms. One such boundary is given by the canonical orthonormal representation of the graph complement of the $8$-vertex graph describing the events occurring in the CHSH inequality as follows \cite{Cabello}:
\begin{eqnarray}
\langle \phi | &=& \left\{ \sqrt{1-\frac{1}{\sqrt{2}}}, \sqrt{1-\frac{1}{\sqrt{2}}}, \sqrt{1-\frac{1}{\sqrt{2}}}, \sqrt{\frac{3}{\sqrt{2}} - 2}, 0 \right\} , \nonumber \\  
\langle w_1 | &=& \left\{1,0,0,0,0 \right\}, \nonumber \\
\langle w_2 | &=& \left\{0,1,0,0,0 \right\}, \nonumber \\
\langle w_3 | &=& \left\{0,0,1,0,0 \right\}, \nonumber \\
\langle w_4 | &=& \left\{2-\sqrt{2},0,0,\sqrt{\sqrt{2}-1},-\sqrt{3\sqrt{2} - 4}\right\}, \nonumber \\
\langle w_5 | &=& \left\{3-2\sqrt{2},2-\sqrt{2},0,\sqrt{2\left(5\sqrt{2}-7 \right)},\sqrt{6\sqrt{2} - 8}\right\}, \nonumber \\
\langle w_6 | &=& \left\{2-\sqrt{2},3-2\sqrt{2},2-\sqrt{2},-2\sqrt{5\sqrt{2}-7},0\right\}, \nonumber \\
\langle w_7 | &=& \left\{0,-2+\sqrt{2},2\sqrt{2}-3,-\sqrt{2\left(5\sqrt{2}-7 \right)},\sqrt{6\sqrt{2} - 8}\right\}, \nonumber \\
\langle w_8 | &=& \left\{0,0,-2+\sqrt{2},-\sqrt{\sqrt{2}-1},-\sqrt{3\sqrt{2}-4} \right\}, \nonumber \\
\end{eqnarray}
with $| w_9 \rangle, \dots, |w_{16} \rangle$ belonging to a subspace orthogonal to $| \phi \rangle$, so that $|\langle \phi | w_i \rangle|^2 = 0$ for $i = 9, \dots, 16$. Moreover, $|\langle \phi | w_i \rangle|^2 = 1- \frac{1}{\sqrt{2}}$ for $i = 1, \dots, 8$. The quantum box leading to maximal violation of the CHSH inequality with $|\mathcal{P} \rangle_i = \frac{2+\sqrt{2}}{8}$ for $i = 1, \dots, 8$ lives on this boundary, and is well-known to self-test the two-qubit singlet state.  

Similarly, the connection between the elliptope and the set of quantum correlations is also useful in deriving novel Quantum Bell inequalities. In particular, the boundary of the quantum correlation set in the CHSH scenario can be generalized to the binary-outcome Bell scenario with $m_A = m_B = m$ inputs on each side. The CHSH Bell inequality directly generalizes to the Braunstein-Caves chain Bell inequalities in this scenario \cite{BC90}. One can parametrize $\langle A_1 B_1 \rangle = \cos\left(\theta_1\right)$, $\langle A_2 B_1 \rangle = \cos\left(\theta_2\right)$, $\langle A_2 B_2 \rangle = \cos\left(\theta_3\right)$, $\langle A_3 B_2 \rangle = \cos\left(\theta_4\right)$, $\dots$, $\langle A_m B_{m} \rangle = \cos\left(\theta_{2m-1}\right)$, $\langle A_1 B_m \rangle = \cos\left(\theta_{2m} \right)$ with $0 \leq \theta_1, \dots, \theta_{2m} \leq \pi$ and where at most one of $\theta_1, \dots, \theta_{2n}$ is greater than $\pi/2$ (yielding the negative sign in the corresponding chain inequality). One can then analogously derive the following boundary of the quantum correlation set in this scenario \cite{BJT93}
\begin{eqnarray}
2 \max_{k \in \{1, \dots, 2m\}} \theta_{k} \leq \sum_{j = 1}^{2m} \theta_j.
\end{eqnarray}
An interesting open question to pursue in future work is whether this boundary of the quantum correlation set is realized by self-testing quantum correlations. 

\subsection{Conclusions}
In this paper, we have shown that that all two-party Bell inequalities that define facets of the classical Bell polytope, are violated in a natural semi-definite programming relaxation to the set of quantum correlations, termed 'Almost Quantum' theory. We have also seen that all correlation Bell inequalities that define facets of the lower dimensional correlation Bell polytope, are violated in quantum theory. Finally, we have shown novel quantum Bell inequalities which should be investigated in future work for self-testing applications \cite{Scarani}. The important open question remains whether every facet-defining Bell inequality (of the classical Bell polytope) is violated in Quantum theory. It would be interesting to see if the methods discussed here can be extended to further levels of the convergent hierarchy of semi-definite programming relaxations of the Quantum set. It would also be interesting to find tight bounds on the dimension of the faces of the quantum set, and information-theoretic explanations behind these.

\subsection{Acknowledgments.}
We acknowledge useful discussions with Stefano Pironio, Pawe{\l} Horodecki and Andreas Winter. This work is supported by the Start-up Fund 'Device-Independent Quantum Communication Networks' from The University of Hong Kong. This work was supported by the National Natural Science Foundation of China through grant 11675136, the Hong Kong Research Grant Council through grant 17300918, and the John Templeton Foundation through grants 60609, Quantum Causal Structures, and 61466, The Quantum Information Structure of Spacetime (qiss.fr). The opinions expressed in this publication are those of the author and do not necessarily reflect the views of the John Templeton Foundation.


\end{document}